\documentclass[]{interact}

\usepackage{epstopdf}
\usepackage[caption=false]{subfig}

\usepackage[longnamesfirst,sort]{natbib}
\bibpunct[, ]{(}{)}{;}{a}{,}{,}
\renewcommand\bibfont{\fontsize{10}{12}\selectfont}

\usepackage{float}

\usepackage{txfonts}
\usepackage{color}
\usepackage{colortbl}
\usepackage{multirow}
\usepackage[table,xcdraw]{xcolor}

\usepackage{enumitem}
\setlist{leftmargin=1.25cm}

\theoremstyle{plain}
\newtheorem{theorem}{Theorem}[section]

\theoremstyle{definition}

\theoremstyle{remark}
\newtheorem{remark}{Remark}

\begin{document}


\title{Filtering and smoothing estimation algorithms from uncertain nonlinear observations with time-correlated additive noise and random deception attacks}

\author{
\name{R. Caballero-{\'A}guila\textsuperscript{a}\thanks{This is an original manuscript of an article published by Taylor \& Francis in INTERNATIONAL JOURNAL OF SYSTEMS SCIENCE on March 19 2024, available at: https://doi.org/10.1080/00207721.2024.2328781.},
J. Hu\textsuperscript{b} and
J. Linares-P{\'e}rez\textsuperscript{c}}
\affil{\textsuperscript{a}Departamento de Estad\'istica. Universidad de Ja{\'e}n.  Paraje Las Lagunillas. 23071. Ja{\'e}n. Spain; \\ \textsuperscript{b}Key Laboratory of Advanced Manufacturing and Intelligent Technology, Ministry of Education, Harbin University of Science and Technology, Harbin 150080, China; \\
\textsuperscript{c}Departamento de Estad\'istica. Universidad de Granada. Avenida Fuentenueva. 18071. Granada. Spain}
}

\maketitle

\begin{abstract}
This paper discusses the problem of estimating a stochastic signal from nonlinear uncertain observations with time-correlated additive noise described by a first-order Markov process. Random deception attacks are assumed to be launched by an adversary, and both this phenomenon and the uncertainty in the observations are modelled by two sets of Bernoulli random variables.  Under the assumption that the evolution model generating the signal to be estimated is unknown and only the mean and covariance functions of the processes involved in the observation equation are available, recursive algorithms based on linear approximations of the real observations are proposed for the least-squares filtering and fixed-point smoothing problems. Finally, the feasibility and effectiveness of the developed estimation algorithms are verified by a numerical simulation example, where the impact of uncertain observation and deception attack probabilities on estimation accuracy is evaluated.
\end{abstract}




\section{Introduction}
\label{sec:Intro}

Although state-space models can theoretically be divided into linear and nonlinear models, in practice there are no strictly linear models. So-called linear systems are nothing more than approximations, usually valid over a limited range of values of the variables involved in the model. Even though many systems have a sufficiently high degree of approximation to linearity, we eventually encounter systems that deviate significantly from linear behaviour, even within a limited operating range. In such cases, the accuracy of linear estimation techniques diminishes and it becomes necessary to explore nonlinear estimation approaches \citep{Simon_2006}. Despite these considerations, there has been a huge amount of work on linear estimation for linear systems and the mathematical tools available for this kind of systems are much more accessible and well understood. For this reason, linear approximations are often used to adapt linear estimation techniques to the nonlinear problems that are encountered in many branches of practical domains, such as computer vision, communications, navigation and tracking systems or econometrics and finance \citep{Hu-Wang-Gao_Springer_2015}.

\vskip.2cm
The performance of estimation algorithms is often affected by the occurrence of random uncertainties, a common one being the presence of missing measurements (also called uncertain observations). Indeed, this phenomenon is unavoidable in many real-world scenarios where the information received by the estimator side is usually incomplete, due to several causes (e.g., random failures in the measurement mechanism, accidental loss of some data packets or inaccessibility of data at certain times). In these situations, it is necessary to consider the influence of this incomplete information when designing estimation algorithms.
In the context of linear systems, the estimation problem for multisensor stochastic uncertain systems with missing measurements and unknown measurement disturbances is addressed in \citet{Pang_Sun_IEEESENSORS_2015}. Using a new augmented state approach, three robust Kalman-like filtering algorithms are proposed in \citet{Ran_Deng_INFSCI_2020} for a class of multisensor systems with mixed uncertainties including random delays, missing measurements, multiplicative
noises and uncertain noise variances. The distributed filtering problem for systems with missing measurements is studied in \citet{Wen_et_al_JFI_2020}, under the assumption that the state noises and the measurement noises are correlated.
In the context of nonlinear systems, the impact of incomplete information on the estimation problem has also been analyzed, e.g., in \citet{Hu_Wang_Caballero_CNSNS_2023}, where a class of singular systems subject to random delays, packet dropouts and nonlinearities,
and in \citet{Han_et_al_2017_NEUROCOMP2017}, where a distributed $H_\infty$-consensus filtering approach is proposed for nonlinear systems with missing measurements. More recently, a particle filter is proposed in \citet{Ma_Wang_CNSNS_2022} for nonlinear systems with time-varying delays and unknown noise distribution and, in \citet{Hu_Hu_caballero_et_al_INFSCI2023}, a distributed resilient fusion filtering algorithm is designed for nonlinear systems with dynamic event-triggered mechanism under the missing-measurement phenomenon. A complete survey on estimation algorithms for nonlinear systems with communication constraints causing random delays, missing/fading measurements or randomly occurring incomplete information can be found in \citet{Hu_Hu_Yang_SSCE_2020}.

\vskip.2cm
Many conventional estimation algorithms are typically based on the assumption that the additive noise in measurements is either white or finite-step correlated. However, in practical engineering applications, infinite-step correlated measurement noises can be prevalent, especially when the sampling frequency is sufficiently high, leading to significant correlation over two or more consecutive sampling periods. Over the past decade, the estimation problem has been widely studied under the assumption that the infinite-step time-correlated noise is a first-order Markov process.
The state estimation problem for linear systems with multiplicative
noise and time-correlated additive noise in the measurements has been investigated in \citet{Liu_IEEETSP_2015} and an improved steady-state filter has been designed in  \citet{Liu_Shi_IEEETAC_2019}.
Least-squares estimation algorithms are proposed in \citet{Garcia_et_al_IJSS2020}, considering random delays in the transmission connections that are modelled by Markov chains, and in \citet{Caballero_SENSORS_2022}, considering random parameter matrices in the measurement model and random packet dropouts for which two different compensation scenarios are compared.
Distributed fusion filtering algorithms for uncertain systems with random delays and packet loss prediction compensation are designed in \citet{Caballero_IJSS_2023}.
Recursive estimation algorithms for linear stochastic uncertain
systems with time-correlated additive noises and packet dropout compensations are proposed in \citet{Ma_Sun_SP2020} and a similar study for nonlinear systems is carried out in \citet{Cheng_IET2021}.

\vskip.2cm
In addressing the estimation problem, security emerges as a crucial consideration that should not be disregarded. The vulnerability to cyber attacks is well-documented in the literature \citep[see, for instance,][]{Mahmoud_2019,Sanchez_2019}. Particularly, the estimation problem in networked systems exposed to deception attacks has been the focus of numerous significant research endeavors.
In general, deception attackers aim to compromise data integrity by maliciously introducing random falsifications into their information. Several research studies have delved into security-guaranteed filtering problems, such as the investigation of centralized solutions for linear time-invariant stochastic systems with multirate-sensor fusion under deception attacks in \citet{Wang_2018_IJFI}. The exploration of the $H_\infty$-consensus filtering problem for discrete-time systems with multiplicative noises and deception attacks is documented in \citet{Han_2019}. Additionally, the study of the chance-constrained $H_\infty$ state estimation problem for a class of time-varying neural networks, subject to measurement degradation and randomly occurring deception attacks, is presented in \citet{Qu_2022}.
Also, the distributed estimation problem in sensor networks  has been addressed under various security threats. Examples include investigations under false data injection attacks in \citet{Yang-et-al-2019-Automatica} and under deception attacks in \citet{Caballero_SENSORS_2023}, \citet{Xiao-et-al-2020}, \citet{Ma-Sun-Sensors_2023} and \citet{Ma-et-al-IEEE_TSIPN_2021}.

\vskip.2cm
Motivated by the preceding discussion, our aim is to address the least-squares (LS) estimation problem of signals using nonlinear uncertain observations (missing measurements) with time-correlated additive noise modeled by a first-order Markov process, in the presence of random deception attacks. The proposed recursive algorithms, based on linear approximations, offer a novel approach to mitigating the impact of missing measurements and deception attacks in signal estimation.
The main contributions of this paper are summarized as follows:
{\em (a)} A covariance-based estimation approach is used, so the evolution model of the signal to be estimated does not need to be known.
{\em (b)} The class of stochastic signals investigated in this paper is quite comprehensive, as the assumptions under which our study is valid are verified by a great variety of stationary and non-stationary signals.
{\em (c)} The direct estimation of the time-correlated additive noise avoids the use of the differencing method or vector augmentation.
{\em (d)} Despite the fact that the simultaneous consideration of uncertain observations, time-correlated noise and random attacks adds complexity to the model, the proposed filtering and fixed-point smoothing algorithms keep the advantages of recursivity and computational simplicity.
{\em (e)} The proposed  estimators have a satisfactory performance even in the presence of high probability of missing measurements and/or high probability of successful random deception attacks.

\vskip .2cm

The paper is structured as follows. Section \ref{sec:Model} details the characteristics of the observation model under consideration. The main results are presented in Section \ref{sec:main results}, which includes the problem statement (subsection \ref{subsec:problem}) and the derivation of the proposed filtering  and fixed-point smoothing algorithms (subsections \ref{subsec:filtering alg} and \ref{subsec:smoothing alg}, respectively). Section \ref{sec:example} conducts a simulation study to showcase the effectiveness of the proposed filtering and fixed-point smoothing estimators, additionally exploring the impact of uncertain observation and random deception attacks probabilities on estimation performance. Finally, Section \ref{sec:conclusions} provides some concluding remarks.

\vskip .2cm

The notations and acronyms that are used throughout the paper are summarized in Table 1.

\begin{table}[hh]
\centering
    \caption{Notations and acronyms used throughout the paper. All vector and matrix dimensions are assumed to be compatible with algebraic operations unless explicitly stated.
}
    \vspace{2mm}
\begin{tabular}{ll}
\hline
$\mathbb{R}^{n}$ & Set of $n$-dimensional real vectors\\
\rowcolor[HTML]{EFEFEF}
$\mathbb{R}^{n\times m}$ & Set of $n\times m$ real matrices\\
$\frac{\partial f(x)}{\partial x}\Big|_{x=a}$ & $m \times n$ Jacobian matrix of $f: \mathbb R^{n} \rightarrow \mathbb R^{m}$ at a point $a \in \mathbb R^{n}$\\
\rowcolor[HTML]{EFEFEF}
$M^T$ and $M^{-1}$  &   Transpose and inverse of   matrix $M$\\
$\delta_{j,i}$ &  Kronecker delta function\\
\rowcolor[HTML]{EFEFEF}
$\mathbb{E}[a]=\overline{a}$ &  Mathematical expectation of a random variable or vector $a$\\
$\Sigma_{k,s}^{{a}}$&  Covariance function of a stochastic process $\{a_k\}_{k\geq 1}$: \\
&  $\Sigma_{k,s}^{a}=Cov[a_k,a_s]=E\big[\big(a_k-\overline{a}_k\big)\big(a_s-\overline{a}_s\big)^T\big], \ \ \Sigma_{k}^{a}=Cov[a_k]$ \\
\rowcolor[HTML]{EFEFEF}
$\widehat{a}_{k/s}$ & Optimal linear estimator of the vector $a_k$ based on $\big\{ y_1,\ldots,y_s\big\}$\\
LS & Least-squares\\
\rowcolor[HTML]{EFEFEF}
OPL & Orthogonal projection lemma\\
EKF & Extended Kalman filter\\
\cline{1-2}
\end{tabular}

\end{table}

\section{Mathematical model and preliminaries}
\label{sec:Model}

\subsection{Signal process}

Consider a signal process $\{x_k\}_{k\geq 1}$ that must be estimated from a set of noisy measurements. Assume that the mathematical model describing the evolution of this signal process is not known, but its mean and covariance functions satisfy the
following hypothesis:

\begin{itemize}
  \item [\em (H1)]  The $n_x$-dimensional signal  $ \{x_k\}_{k \geq 1}$ is a second-order  zero-mean random process and its
covariance function, $\Sigma_{k,s}^x = Cov[x_{k}, x_s]$, can be factorized as
$$ \Sigma_{k,s}^x =   A_k  B^T_s,\  s\leq k, $$ where $  A_k,  B_s \in
\mathbb{R}^{n_x \times p}$ are known  matrices.

\end{itemize}

\begin{remark}
  Hypothesis \emph{(H1)} on the signal covariance function is verified by a great variety of stationary and non-stationary signals \citep{Caballero_SENSORS_2022}. Estimation approaches based on this hypothesis, rather than the state-space model, thus provide a comprehensive framework to obtain general algorithms that cover a wide range of practical situations.
\end{remark}

\subsection{Nonlinear uncertain measurements with time-correlated additive noise}

The signal estimation will be performed from $n_z$-dimensional nonlinear outputs that are perturbed by additive noise and subject to random failures, causing that some measured outputs are only noise. This phenomenon is described by a sequence of Bernoulli random variables, $ \{\gamma_k\}_{k \geq 1}$. When $\gamma_k=1$, the actual measurement value is equal to the original measurement value, while $\gamma_k=0$ means that the actual measurement is only noise. More specifically, the actual measurements are described by the following mathematical model:

\begin{equation}\label{actual_measurements}
  z_k=\gamma_k h_k(x_k)+v_k, \ \ k\geq 1,
\end{equation}

\noindent with the following hypotheses:

\begin{itemize}
   \item [\em (H2)] $ \big\{\gamma_k\big\}_{ k\geq 1}$ is a sequence of  independent  Bernoulli random variables with  known mean function $\overline{\gamma}_k=\mathbb{E}[\gamma_k], \  k\geq 1.$

   \item [\em (H3)] For all $ k\geq1$, the function $h_k: \mathbb R^{n_x} \rightarrow \mathbb R^{n_z}$ is an analytic function.
\end{itemize}

\begin{remark}
Hypothesis \emph{(H3)} guarantees that $h_k$ is infinitely differentiable and,  for every $x_0 \in \Bbb R^{n_x}$,  its Taylor series about $x_0$ converges to the function in some neighborhood of $x_0$. Typical examples of analytic functions are: polynomials, exponential functions, logarithmic functions, trigonometric functions and power functions \citep{Krantz_2002}.
\end{remark}

\vskip .15cm

In many engineering applications,  the additive noise perturbing the observations has an infinite-step correlation; to better model such situations, the following first-order Markov model is used:

\begin{equation}\label{time-correlated_noise}
v_k=D_{k-1}v_{k-1}+u_{k-1}, \ \ k\geq 1,
\end{equation}

\noindent where $D_k$ is non-singular for all $k\geq 0$ and the following hypotheses hold:

\begin{itemize}
   \item [\em (H4)] $v_{0}$ is a zero-mean random vector with known covariance $\Sigma^v_{0}=Cov[v_{0}].$

   \item [\em (H5)] The noise process $\{u_{k}\}_{k\geq 0}$ is a zero-mean white sequence with known covariance function $\Sigma^{u}_{k}=Cov[u_{k}], \  k\geq 0$, and it  is independent of the initial vector $v_{0}$.
\end{itemize}

\begin{remark}
Hypotheses \emph{(H4)} and \emph{(H5)}, together with the non-singularity of $D_k$, guarantee that the covariance function of the time-correlated additive noise, $\Sigma_{k,s}^v = Cov[v_{k}, v_s]$, admits the following factorization:
\begin{equation*}\label{Noise covariance factorization}
  \Sigma_{k,s}^v   = E_k F_s^T , \ \ \ s \leq k,
\end{equation*}
where $E_k=D_{k-1} \cdots D_{0} $, $F_s= \Sigma_{s}^v (E_s^{-1})^T$ and $\Sigma_{s}^v= Cov[v_{s}]$ is recursively obtained as $$\Sigma_s^v = D_{s-1} \Sigma_{s-1}^v D_{s-1}^T + \Sigma^{u}_{s-1},\ s\geq 1.$$
\end{remark}

\subsection{Random deception attacks}

In many practical situations, a crucial issue that cannot be ignored in the study of the estimation problem is the possible occurrence of cyber-attacks.
In particular, deception attacks constitute a significant threat that attempts to compromise data integrity by maliciously and randomly falsifying information. In this type of attacks, the signal injected by the attacker, $\breve{z}_k$, aims to neutralise the actual measurement, $z_k$, and replace it with a deceptive noise, $w_k$. Specifically,
\[\breve{z}_k=-z_k+w_k, \ \ k\geq 1,\]
where
\begin{itemize}
   \item [\em (H6)]  The noise process $\{w_{k}\}_{k\geq 1}$ is a zero-mean white sequence with known covariance function $\Sigma^{w}_{k}=Cov[w_{k}], \  k\geq 1.$
\end{itemize}

To describe the fact that, in practice, cyber-attacks are often unpredictable and random, a sequence of Bernoulli random variables is adopted. More specifically, the following model with stochastic deception attacks is considered to describe the measurements processed for the estimation:
\[y_k=z_k+\lambda_k\breve{z}_k, \ \ k\geq 1,\]
which can be equivalently rewritten as
\begin{equation}\label{available_measurements}
y_k=(1-\lambda_k)z_k+\lambda_k w_k, \ \ k\geq 1,
\end{equation}
\noindent where

\begin{itemize}
      \item [\em (H7)] $ \big\{\lambda_k\big\}_{ k\geq 1}$   is a sequence of  independent  Bernoulli random variables with  known mean function $\overline{\lambda}_k=\mathbb{E}[\lambda_k], \  k\geq 1.$
\end{itemize}

\begin{remark}The binary values of $\lambda_k$ indicate if the adversary actually launch the attack ($\lambda_k=1$) or not ($\lambda_k=0$). When the actual measurement is  attacked, $y_k=w_k$ and only noise will be processed. Otherwise, if no attack is injected, $y_k=z_k$ and the actual measurement is processed for the estimation.
\end{remark}

\vskip.15cm

Finally, the following independence hypotheses on the processes involved in the considered model is required to
address the estimation problem.

\begin{itemize}
   \item [\em (H8)] The signal process $\{x_k\}_{ k\geq1}$ and the processes  $\{\gamma_k\}_{ k\geq1}$, $\{v_k\}_{ k\geq1}$,
  $\{w_k\}_{ k\geq 1}$  and $\{\lambda_k\}_{ k\geq1}$
      are mutually independent.
\end{itemize}

\subsection{Linearized observations}

Under hypotheses \emph{(H1)--(H8)}, our goal is to obtain recursive algorithms for the filtering and fixed-point smoothing problems; that is, to address  the estimation problem of the signal $x_k$ given the nonlinear observations $\{y_1, \ldots , y_L\}, \ L\geq k$. To this end, we will use a similar reasoning to that used to derive the extended Kalman filter. The problem is thus reduced to linearizing the observation equation (\ref{actual_measurements}) and then inserting such linearized observations into (\ref{available_measurements}) to calculate the LS linear estimator of the signal from the resulting measurements.

More precisely, from hypothesis \emph{(H3)}, assuming knowledge of a nominal trajectory of the
signal, $\{x_k^0\}_{k\geq 1}$, the function $h_k$  can be expanded in
Taylor series about $x_k^0$:

$$h_k(x)=h_k(x_k^0)+ \frac{\partial h_k(x)}{\partial x}\Big|_{x=x_k^0}(x-x_k^0)+ \cdots, \ \ k\geq 1,$$

Neglecting the terms of order greater than one in this Taylor expansion, equation (\ref{actual_measurements}) can be approximated by the following \emph{linearized observation equation}
\begin{equation}\label{linearized actual_measurements}
  z_k=\gamma_k (H_k x_k + C_k) +v_k, \ \ k\geq 1,
\end{equation}
where
$$H_k=\displaystyle \frac{\partial h_k(x)}{\partial x}\Big|_{x=x_k^0}, \ \quad C_k=h_k(x_k^0)-H_kx_k^0.$$

The above equation is a linear equation affected by the binary multiplicative noise, $\{\gamma_k\}_{k\geq 1}$, and the time-correlated additive noise, $\{v_k\}_{k\geq 1}$. Our aim is to derive recursive estimation algorithms by using the linear approximation (\ref{linearized actual_measurements}) in the equation for the available observations (\ref{available_measurements}).

\section{Main results}
\label{sec:main results}

\subsection{Problem statement and innovation approach}
\label{subsec:problem}

Our goal  is to obtain recursive algorithms for the LS linear filter and fixed-point smoother of the signal $x_k$ from the available observations (\ref{available_measurements}), based on the linearized measurements (\ref{linearized actual_measurements}).
For this purpose, we will use an \emph{innovation approach}. According to this approach, the  observation
process  is transformed into an equivalent innovation process and the LS linear estimator,
$\widehat{\zeta}_{k/L}$, of a random vector $\zeta_k$ based on a set
of observations $\left\{y_j, \ 1\leq j\leq L\right\}$, can be
expressed  as a linear combination of the innovations
as follows:
\begin{equation}\label{general expression estimators}
  \widehat{\zeta}_{k/L}=\displaystyle{\sum_{j=1}^{L}}
 {\mathcal{S}}^{\zeta}_{k,j} (\Sigma_{j}^{\eta})^{-1} \eta_j, \  \  k,L \geq 1,
\end{equation}
where   ${\mathcal{S}}^{\zeta}_{k,j}=\mathbb{E}[\zeta_k\eta^{T}_j]$,  $\eta_j= y_j- \widehat{y}_{j/j-1}$ is the innovation at time $j$ and $\Sigma^{\eta}_{j}=Cov[\eta_j]$.

Using (\ref{available_measurements}) and the model hypotheses, the innovation can be written as
 \begin{equation}\label{inno}
   \eta_k = y_k-(1-\overline{\lambda}_k)\widehat{z}_{k/k-1}, \  \  k \geq 1,
 \end{equation}
and, from the linear approximation (\ref{linearized actual_measurements}), it is clear that the prediction estimator of $z_k$ can be approximated by
\begin{equation}\label{pred_z}
 \widehat{z}_{k/k-1} = \overline{\gamma}_k\big(H_k \widehat{x}_{k/k-1}+C_k\big) + \widehat{v}_{k/k-1}.
\end{equation}

\vskip.15cm
Expressions (\ref{general expression estimators})--(\ref{pred_z}) are the key points for the derivation of the recursive filtering and fixed-point smoothing algorithms that will be presented in the following subsections.

\subsection{Recursive filtering algorithm}
\label{subsec:filtering alg}

For the sake of convenience, we introduce the following notations:
 \begin{equation}\label{Gamma}
    \Gamma^{a}_k = \left\{\begin{array}{ll}
                           \overline{\gamma}_kH_k B_k,& \ a=x, \\
                            F_k, & \ a=v,
                          \end{array}
    \right.
  \end{equation}
 \begin{equation}\label{Delta}
    \Delta^{a}_k = \left\{\begin{array}{ll}
                           \overline{\gamma}_kH_k A_k,& \ a=x, \\
                            E_k, & \ a=v.
                          \end{array}
    \right.
  \end{equation}

\begin{theorem}\label{Th_filter}
  Consider the observation model (\ref{actual_measurements})--(\ref{available_measurements}), where the processes involved satisfy hypotheses (H1)--(H8). Then,
the innovation $\eta_k$ is calculated as
  \begin{equation}\label{eta_innov_algorithm}
    \eta_k = y_k - (1-\overline{\lambda}_k)\big(\Delta^{x}_k e^{x}_{k-1} + \Delta^{v}_k e^{v}_{k-1}+\overline{\gamma}_k C_k \big), \ \ k\geq 1,
  \end{equation}
  and its covariance matrix $\Sigma^{\eta}_k$ satisfies
  \begin{equation}\label{Xi_Covarianza innov}
    \begin{array}{ll}
      \Sigma^{\eta}_k= &  \Sigma^{y}_k - (1-\overline{\lambda}_k)^2\Big[\left(
      \Delta_k^x T_{k-1}^{xx}  + \Delta_k^v T_{k-1}^{vx}  \right) (\Delta_k^{x})^T\\
   & + \left(
      \Delta_k^x T_{k-1}^{xv}  + \Delta_k^v T_{k-1}^{vv}  \right) (\Delta_k^{v})^T  +\overline{\gamma}^2_k C_k C_k^{T} \Big],  \ \ k\geq  1,
    \end{array}
  \end{equation}
  with
  \begin{equation}\label{Covarianzas y z}
    \begin{array}{l}
    \Sigma^{y}_k =(1-\overline{\lambda}_k)\Sigma^{z}_{k}+\overline{\lambda}_k \Sigma^{w}_{k},\quad k\geq 1, \\
    \Sigma^{z}_{k}= \overline{\gamma}_k\left(H_{k}A_kB^{T}_kH^{T}_{k}+  C_k C_k^{T} \right)+E_k F_k^T ,\quad k\geq 1.
    \end{array}
  \end{equation}
   The vectors $e^a_k$ ($a=x,v$) are recursively obtained by
   \begin{equation}\label{e_a_k recursion}
    e^{a}_k=e^{a}_{k-1}+\Psi^{a}_k(\Sigma^{\eta}_k)^{-1}\eta_k,\ \ k\geq 1; \ \ e^{a}_0=0,
  \end{equation}
  where
  \begin{equation}\label{Psi}
    \Psi^{a}_k=
(1-\overline{\lambda}_k)\big(\Gamma_k^a - \Delta^{x}_k T^{xa}_{k-1} - \Delta^{v}_k T^{va}_{k-1} \big)^T , \ \ k\geq 1.
  \end{equation}
  The matrices $T^{ab}_k = \mathbb{E}[e^a_k (e^b_k)^T]$ ($a,b=x,v$) are also recursively calculated by
   \begin{equation}\label{T_ab_k_recursion}
    T^{ab}_k=T^{ab}_{k-1}+\Psi^{a}_k(\Sigma^{\eta}_k)^{-1}(\Psi^{b}_k)^T, \ \ k\geq 1; \ \ T^{ab}_0=0.
  \end{equation}
  The filtering estimator of the signal,
  $\widehat{x}_{k/k}$,
  is then computed by
  \begin{equation}\label{filter}
    \widehat{x}_{k/k}=A_k{e^{x}_k}, \ \ k\geq 1.
  \end{equation}

\end{theorem}

\begin{proof}
  According to (\ref{general expression estimators}), the estimator of the signal $x_k$ based on a set of observations $\left\{y_1, \ldots, y_L \right\}$, with $L \leq k$, is given by
  \begin{equation*}\label{general expression estimators x}
  \widehat{x}_{k/L}=\displaystyle{\sum_{j=1}^{L}}
 {\mathcal{S}}^{x}_{k,j} (\Sigma_{j}^{\eta})^{-1} \eta_j, \  \  L \leq k,
\end{equation*}
where ${\mathcal{S}}^{x}_{k,j}=\mathbb{E}[x_k\eta^{T}_j]$.
From (\ref{inno}), it is clear that
$$
\mathcal{S}^{x}_{k,j}= \mathbb{E}\big[x_k y_{j}^T\big]-(1-\overline{\lambda}_j)\mathbb{E}\big[x_k\widehat{z}^T_{j/j-1}\big].
$$
Using (\ref{available_measurements}) and (\ref{linearized actual_measurements}), and taking into account hypotheses \emph{(H1)} and \emph{(H8)}, we have
$$\mathbb{E}\big[x_ky^{T}_{j}\big]= (1-\overline{\lambda}_j)\overline{\gamma}_j A_kB_j^{T} H_j^T, \ \ 1\leq j \leq k,$$
and, from (\ref{pred_z}), the following expression is deduced
$$\mathbb{E}\big[x_k\widehat{z}^T_{j/j-1}\big]=
\overline{\gamma}_j\mathbb{E}\big[ x_k\widehat{x}^T_{j/j-1}\big]H_j^T + \mathbb{E}\big[x_k\widehat{v}^{T}_{j/j-1}\big], \ \ 1\leq j \leq k.$$
Using again the general expression (\ref{general expression estimators}) for the estimators $\widehat{x}^T_{j/j-1}$ and $\widehat{v}^T_{j/j-1}$, we can write
$$
\mathbb{E}\big[
 x_k\widehat{a}^{T}_{j/j-1}\big]
   =\displaystyle{\sum_{i=1}^{j-1}}{\mathcal S}^{x}_{k,i}(\Sigma^{\eta}_i)^{-1} ({\mathcal S}_{j,i}^{a})^T, \ \ j\geq 2, \ \ a=x,v.
$$

Consequently, $\mathcal{S}^{x}_{k,j}$ admits the following expression
$$
{\mathcal S}^{x}_{k,j}\hskip -0.05cm =
(1-\overline{\lambda}_j)\Big[ \ \overline{\gamma}_j A_kB_j^T H_j^T -(1-\delta_{j,1})\displaystyle{\sum_{i=1}^{j-1}}{\mathcal S}^{x}_{k,i}(\Sigma^{\eta}_i)^{-1} \big( \overline{\gamma}_j H_j{\mathcal S}^{x}_{j,i} + {\mathcal S}_{j,i}^{{v}} \big)^T \Big],\  \ j\geq 1,
$$
from which the following factorization is easily deduced:
\begin{equation}\label{S_x_kj}
{\mathcal S}^{x}_{k,j}= A_k\Psi^{x}_j, \ \ j\leq k,
\end{equation}
just defining
\begin{equation}\label{Psi_x_j}
\Psi^{x}_j=(1-\overline{\lambda}_j)\Big[\ \overline{\gamma}_j B_j^TH_j^T -(1-\delta_{j,1})\displaystyle{\sum_{i=1}^{j-1}}\Psi^{x}_{i}(\Sigma^{\eta}_i)^{-1} \big( \overline{\gamma}_jH_jA_j\Psi^{x}_{i} + {\mathcal S}_{j,i}^{v} \big)^T \Big] ,\  \ j\geq 1.
\end{equation}
Similarly, the following identity is derived
\begin{equation}\label{S_v_kj}
{\mathcal S}^{v}_{k,j}= E_k\Psi^{v}_j, \ \ j\leq k,
\end{equation}
with
\begin{equation}\label{Psi_v_j}
\Psi^{v}_j=(1-\overline{\lambda}_j)\Big[\ F_j^T  -(1-\delta_{j,1})\displaystyle{\sum_{i=1}^{j-1}}\Psi^{v}_{i}(\Sigma^{\eta}_i)^{-1}  \big(\overline{\gamma}_j H_jA_j\Psi^{x}_{i} + E_j \Psi^{v}_i \big)^T \Big] ,\  \ j\geq 1.
\end{equation}
Thus, defining the following vectors:
\begin{equation}\label{e_a_k}
e^{a}_k = \displaystyle{\sum_{j=1}^{k}}\Psi^{a}_j(\Sigma^{\eta}_j)^{-1}\eta_j,\ \ k\geq 1, \ \ a=x,v,
\end{equation}
and using (\ref{S_x_kj}) and (\ref{S_v_kj}), we have that
\begin{equation}\label{a_pred_fil}
 \widehat{x}_{k/k}=A_k e^{x}_{k}, \ \ \ \ \   \widehat{x}_{k/k-1}= A_k e^{x}_{k-1},  \ \ \ \ \     \widehat{v}_{k/k-1}=E_k e^{v}_{k-1},  \ \ \ \ \  k\geq 1.
\end{equation}
Substituting these expressions for $\widehat{x}_{k/k-1}$ and $\widehat{v}_{k/k-1}$ in (\ref{pred_z}) and taking into account the definition of $\Delta_k^a$ ($a=x,v$) given in (\ref{Delta}) we have that
\begin{equation}\label{pred_z_proof}
\widehat{z}_{k/k-1} = \Delta_k^x {e^{x}_{k-1}} +  \Delta_k^v {e^{v}_{k-1}}+ \overline{\gamma}_k C_k,\ \ k\geq 1.
\end{equation}
So, using  (\ref{inno}), expression (\ref{eta_innov_algorithm}) for the innovation is straightforward.

In order to obtain the innovation covariance matrix, $\Sigma^{\eta}_k = \mathbb{E}\big[\eta_k \eta_k^T\big]$,  we use (\ref{inno}) and the OPL to deduce that
$$\Sigma^{\eta}_k= \Sigma^{y}_k - (1-\overline{\lambda}_k)^2 \mathbb{E}\big[ \widehat{z}_{k/k-1} \widehat{z}_{k/k-1}^T \big],  \ k\geq  1.$$
Defining $T^{ab}_k = \mathbb{E}\big[e^{a}_k (e^{b}_k)^{T}\big],\ \ k\geq 1$ ($a,b=x,v$), and using (\ref{pred_z_proof}), we can write
$$
\begin{array}{ll}
  \mathbb{E}\big[ \widehat{z}_{k/k-1} \widehat{z}_{k/k-1}^T \big]= & \Delta_k^x T_{k-1}^{xx} (\Delta_k^{x})^T + \Delta_k^x T_{k-1}^{xv} (\Delta_k^{v})^T \\
  & +\Delta_k^v T_{k-1}^{vx} (\Delta_k^{x})^T + \Delta_k^v T_{k-1}^{vv} (\Delta_k^{v})^T +\overline{\gamma}^2_k C_k C_k^{T}
\end{array}
$$
and expression (\ref{Xi_Covarianza innov}) for the innovation covariance matrix is immediately obtained.
The formulas for $\Sigma_k^y$ and $\Sigma_k^z$ given in (\ref{Covarianzas y z}) are easily derived from the model hypotheses.

Using (\ref{e_a_k}), the recursion (\ref{e_a_k recursion}) is directly obtained and, also, the following expression for  $T^{ab}_k$ is straightforward:
\begin{equation}\label{T_ab_k}
  T^{ab}_k = \displaystyle{\sum_{j=1}^{k}}\Psi^{a}_j(\Sigma^{\eta}_j)^{-1}(\Psi^{b}_j)^T,\ \ k\geq 1,
\end{equation}
which, together with the definitions (\ref{Gamma}) and (\ref{Delta}), leads to expression (\ref{Psi}) just by substitution in (\ref{Psi_x_j}) and (\ref{Psi_v_j}). Also, from the above expression, the recursion (\ref{T_ab_k_recursion}) is immediately obtained. Finally, the filter expression (\ref{filter}) has been already obtained in (\ref{a_pred_fil}), so the proof is complete.
\end{proof}

\begin{remark}
The derivation of the estimation algorithm presented in Theorem \ref{Th_filter} is based on the linear approximation (\ref{linearized actual_measurements}) of the nonlinear observations around a nominal trajectory of the signal. Hence, the first question that arises is what to do if we do not have a reliable nominal signal trajectory. In such cases, we can follow an EKF approach; in other words, we can use the prediction estimates $\{\widehat{x}_{k/k-1}\}_{k\geq 1}$ as a nominal trajectory of the signal, because $\widehat{x}_{k/k-1}$ is our best approximation of $x_k$ before the observation at time $k$ is considered.
When doing so, the matrices $H_k$ and $C_k$ in equation (\ref{linearized actual_measurements}) are given by
$$H_k=\displaystyle \frac{\partial h_k(x)}{\partial x}\Big|_{x=A_k e^{x}_{k-1}}, \ \quad C_k=h_k(A_k e^{x}_{k-1})-H_k A_k e^{x}_{k-1},$$
since the prediction estimate at time $k$ is calculated as $\widehat{x}_{k/k-1}=A_k{e^{x}_{k-1}}$ (see (\ref{a_pred_fil})).

Substituting these matrices in (\ref{pred_z}) and taking into account that $\widehat{v}_{k/k-1}=E_k e^{v}_{k-1}$ (see (\ref{a_pred_fil})),
we obtain that
\begin{equation}\label{PRED_predictor_medida_real}
    \widehat{z}_{k/k-1}=\overline{\gamma}_kh_k(A_k e^{x}_{k-1})+E_k e^{v}_{k-1}, \ \ k\geq 1,
\end{equation}
from which expression (\ref{eta_innov_algorithm}) for the innovation  admits the following simplified form:
\begin{equation}\label{PRED_innovation}
    \eta_k=y_k-(1-\overline{\lambda}_k)\left( \overline{\gamma}_kh_k(A_k e^{x}_{k-1})+E_k e^{v}_{k-1}\right), \ \ k\geq 1,
\end{equation}
and  the innovation covariance matrix,  $\Sigma^{\eta}_k$, can be written as
\begin{equation}\label{PRED_Xi_Covarianza innov}
  \Sigma^{\eta}_k= (1-\overline{\lambda}_k)\Sigma^{\widetilde{z}}_{k/k-1}+\overline{\lambda}_k(1-\overline{\lambda}_k)\widehat{z}_{k/k-1}\widehat{z}^T_{k/k-1}
    +\overline{\lambda}_k\Sigma^{w}_{k}, \ \ k\geq 1,
  \end{equation}
whith
\begin{equation}\label{error_cov}
  \begin{array}{ll}
    \Sigma^{\widetilde{z}}_{k/k-1}= &\overline{\gamma}_kH_k\Sigma^{\widetilde{x}}_{k/k-1}H^T_k +\overline{\gamma}_k(1-\overline{\gamma}_k)h_k(A_k e^{x}_{k-1})h^T_k(A_k e^{x}_{k-1})\\
                                    & +\Sigma^{\widetilde{v}}_{k/k-1}, \ \ k\geq 1,\\
    \Sigma^{\widetilde{x}}_{k/k-1}= &A_kB_k^T - A_k T_{k-1}^{xx} A_k^T, \ \ k\geq 1,\\
    \Sigma^{\widetilde{v}}_{k/k-1}= &E_kF_k^T - E_k T_{k-1}^{vv} E_k^T, \ \ k\geq 1.\\

  \end{array}
\end{equation}
\end{remark}

\vskip.15cm
The following steps summarize the filtering algorithm and its computational procedure when the prediction estimates, $\widehat{x}_{k/k-1}=A_k{e^{x}_{k-1}}$, are used as a nominal trajectory of the signal.

\vskip .25cm

\noindent \emph{Filtering algorithm using $x_k^0=A_k e^{x}_{k-1}$ as a  nominal trajectory of the signal} \vskip .15cm

\begin{itemize}
\item[\emph{Step 1.}] Set $k=1$ and initialize the algorithm with $e^a_0=0$ and $T_0^{ab}=0$ ($a,b=x,v$).
\item[\emph{Step 2.}] Compute $H_k=\displaystyle \frac{\partial h_k(x)}{\partial x}\Big|_{x=A_k e^{x}_{k-1}}$ and, from it, compute $\Gamma^{a}_k$ and $\Delta^{a}_k$ ($a=x,v$) by (\ref{Gamma}) and (\ref{Delta}), respectively.
\item[\emph{Step 3.}] Compute $\Psi_k^{a}$ ($a=x,v$)  by (\ref{Psi}).
\item[\emph{Step 4.}] Compute $\widehat{z}_{k/k-1}$ by (\ref{PRED_predictor_medida_real})
        and, from it, compute the innovation $\eta_k$ by (\ref{PRED_innovation}).
\item[\emph{Step 5.}] Compute the matrices $\Sigma^{\widetilde{z}}_{k/k-1}$ by (\ref{error_cov})
    and, from them, compute the innovation covariance matrix $\Sigma^{\eta}_k$ by (\ref{PRED_Xi_Covarianza innov}).
\item[\emph{Step 6.}] Compute $e^a_k$ ($a=x,v$) by (\ref{e_a_k recursion}) and $T^{ab}_k$ ($a,b=x,v$)  by (\ref{T_ab_k_recursion}).
\item[\emph{Step 7.}] Compute the filter, $\widehat{x}_{k/k}$,
by (\ref{filter}).
\item[\emph{Step 8.}] Set $k=k+1$ and return to \emph{Step 2}.
\end{itemize}

\subsection{Recursive fixed-point smoothing algorithm}
\label{subsec:smoothing alg}

The following algorithm allows us to update the filter at any time $k$ as the measurements keep rolling in. More precisely, it allows us to obtain the fixed-point smoothing estimators $\widehat{x}_{k/k+1}, \widehat{x}_{k/k+2}, \ldots$.

\begin{theorem}\label{Th_smoother}

 Starting from the filter, $\widehat{x}_{k/k}$, at a fixed sampling time $k\geq 1$, the fixed-point smoothers, $\widehat{x}_{k/L}$, $L>k$, satisfy the following recursion:
  \begin{equation}\label{smoother_recursion}
   \widehat{x}_{k/L}=\widehat{x}_{k/L-1}+ {\mathcal{S}}^{x}_{k,L} (\Sigma^{\eta}_L)^{-1} \eta_{L}, \  L>k,
 \end{equation}
 with
\begin{equation}\label{S_x_kL}
  {\mathcal{S}}^{x}_{k,L}=  (1-\overline{\lambda}_{L}) \Big[\left(B_k - M_{k,L}^x \right)(\Delta_L^x)^T - M_{k,L}^v (\Delta_L^v)^T \Big], \  L>k.
\end{equation}
The matrices $ M_{k,L}^a=\mathbb{E}[x_k(e^{a}_{L})^T]$, $a=x,v$ are recursively calculated by
\begin{equation}\label{M_a_kL}
  M_{k,L}^a = M_{k,L-1}^a + {\mathcal{S}}^{x}_{k,L}(\Sigma^{\eta}_L)^{-1} (\Psi_L^{a})^T, \ L>k; \qquad  M_{k,k}^a = A_k T_k^{xa}
\end{equation}

\end{theorem}

\begin{proof}
The recursion (\ref{smoother_recursion}) is immediately deduced from (\ref{general expression estimators}), so we must find an expression for ${\mathcal{S}}^{x}_{k,L}=\mathbb{E}[x_k\eta^{T}_{L}]$, $L>k$. A similar reasoning to that used to obtain
${\mathcal{S}}^{x}_{k,j}$ ($j\leq k$) in Theorem \ref{Th_filter} yields
$$
{\mathcal S}^{x}_{k,L} =
(1-\overline{\lambda}_L)\Big[ \ \overline{\gamma}_L B_kA_L^T H_L^T -\displaystyle{\sum_{i=1}^{L-1}}{\mathcal S}^{x}_{k,i}(\Sigma^{\eta}_i)^{-1} \big( \overline{\gamma}_L H_L{\mathcal S}^{x}_{L,i} + {\mathcal S}_{L,i}^{{v}} \big)^T \Big],\  \ L >k.
$$
Taking into account that, from (\ref{S_x_kj}) and (\ref{S_v_kj}), ${\mathcal S}^{x}_{L,i}= A_L \Psi_i^{x}$ and ${\mathcal S}^{v}_{L,i}= E_L \Psi_i^{v}$, respectively, and using (\ref{Delta}), the above expression can be rewritten as
$$
{\mathcal S}^{x}_{k,L} =
(1-\overline{\lambda}_L)\Big[ \  B_k (\Delta_L^x)^T -\displaystyle{\sum_{i=1}^{L-1}}{\mathcal S}^{x}_{k,i}(\Sigma^{\eta}_i)^{-1} \big( \Delta_L^x \Psi_i^{x}+ \Delta_L^v \Psi_i^{v} \big)^T \Big],\  \ L >k.
$$
Then, defining $M_{k,L}^a= \mathbb{E}[x_k(e^{a}_{L})^T]$ and using (\ref{e_a_k}), we have that
  $$M_{k,L}^a=\displaystyle{\sum_{i=1}^{L}}{\mathcal S}^{x}_{k,i}(\Sigma^{\eta}_i)^{-1}(\Psi^{a}_{i})^{T}, \ a=x,v,$$
from which the formula (\ref{S_x_kL}) for ${\mathcal{S}}^{x}_{k,L}$ is straightforward and also the recursion  (\ref{M_a_kL}) is immediately deduced.  Its initial condition is derived just using that, for $i\leq k$, ${\mathcal S}^{x}_{k,i}= A_k \Psi_i^{x}$ and taking into account expression (\ref{T_ab_k}). The proof is then complete.

\end{proof}

\section{Numerical simulation example}
\label{sec:example}

In this section, a  simple numerical simulation example concerning the phase
modulation of analogue communication systems is presented with a dual purpose:
on the one hand, to illustrate the implementation and performance of the proposed filtering and  fixed-point smoothing algorithms and, on the other hand,  to analyze the effect of missing measurements and deception attacks
 on the estimation accuracy.
\vskip .25cm
\noindent {\em Scalar signal process: AR(2) model.}  Consider that  the
 modulating signal $ \{x_k\}_{k \geq 1}$ is a stationary stochastic process with
 autocovariance function
$$ \Sigma^x_{k,s}=Q_1\beta_1^{k-s}+Q_2\beta_2^{k-s},\quad k,s\geq 1,$$
where, for $b_1=0.1$, $b_2=-0.5$ and $\sigma^2=0.25$, the values of $\beta_i$ and $Q_i$, $i=1,2$, are given by:
{\small $$\beta_1,\beta_2=\dfrac{-b_1\pm \sqrt{b_1^2-4b_2}}2,\quad Q_1=\displaystyle
\frac{\sigma^2\beta_1\left(\beta_2^2-1\right)}{\left(\beta_2-\beta_1\right)\left(\beta_1\beta_2+1\right)},\quad
Q_2=-\displaystyle
\frac{\sigma^2\beta_2\left(\beta_1^2-1\right)}
{\left(\beta_2-\beta_1\right)\left(\beta_1\beta_2+1\right)}.
$$}
According to hypothesis {\em (H1)}, this autocovariance
function can be expressed in a separable form defining, for example,
$A_k=\left[Q_1\beta_1^k \  \ Q_2\beta_2^k\right]$ and
$B_k=\left[\beta_1^{-k} \   \  \beta_2^{-k} \right]$.
\vskip .15cm
For simulation purposes, the signal is assumed to be generated from the following second-order
autoregressive model:
$$x_k = -b_1x_{k-1}-b_2x_{k-2} + \varepsilon_k,\ k\geq3; \ \ \ x_2 = -b_1x_{1}+\varepsilon_2; \ \ \ x_1=\varepsilon_1,$$
where $ \{\varepsilon_k\}_{k \geq 1}$ is a zero-mean white Gaussian noise with variance $\Sigma^\varepsilon_{k} = \sigma^2,  \forall k$.

\vskip .25cm
\noindent{\em Uncertain measurements of the carrier signal.}
 Consider the following scalar carrier signal:
$$h_k(x_k)=cos\big(2\pi f_p k\Delta+m_A
x_k\big), \ \ k\geq 1,$$
 where $f_p= 10[Hz]$ is the carrier frequency,
$\Delta=0.01$ is the sampling period of the modulating signal $x_k$
and $m_A=2$ represents the phase sensitivity.
Clearly, if we use the prediction estimates  as a nominal trajectory of the signal, the functions $H_k$ in equation (\ref{linearized actual_measurements}) are
$$H_k=-m_A sin\big(2\pi f_p k\Delta+m_A A_k{e^{x}_{k-1}}\big), \ \ k\geq 1.$$

\vskip .25cm
According to the theoretical model, let us suppose that the  measurements of the carrier signal $h_k(x_k)$ are
given by (\ref{actual_measurements}), where:
\begin{itemize}
 \item [$-$] $\big\{\gamma_k\big\}_{ k\geq 1}$ is a sequence of  independent  Bernoulli random variables with  probabilities $P\big(\gamma_k=1\big)=\mathbb{E}[\gamma_k]=\overline{\gamma}, \  k\geq 1$.
\item [$-$] The noise process  $\{v_{k}\}_{k\geq 0}$ is generated by (\ref{time-correlated_noise}),  where  $D_k=0.75$,  $\{u_k\}_{ k\geq0}$ is a zero-mean white Gaussian noise with $\Sigma^{u}_{k} =0.01, \ \forall k\geq 0$, and  $v_0$ is a zero-mean Gaussian variable with $\Sigma^v_{0}=0.1$.
\end{itemize}

\vskip .25cm

\noindent{\em Random deception attacks.} Finally,  also according to the theoretical  model, let us suppose that the measurements
are subject to deception attacks and the observations available  for the estimation,  $y_k$,  are given by (\ref{available_measurements}), where:

\begin{itemize}
 \item [$-$]  The noise of the false data injection attacks $\{w_{k}\}_{k\geq 1}$ is a standard Gaussian white process.
  \item [$-$] The status of the attacks is described by a white sequence of Bernoulli random variables  $\{\lambda_k\}_{ k\geq 1}$,  with probabilities $P(\lambda_k=1)=\mathbb{E}[\lambda_k]=\overline{\lambda}, \  k\geq 1.$
\end{itemize}

 Under these conditions, using the proposed filtering and smoothing algorithms, the phase demodulation problem is considered. This problem consists of estimating the signal $x_k$ from the observed values $y_k$ and, for this purpose, we have implemented a MATLAB program that simulates the values of the signal, $x_k$, the uncertain measurements, $z_k$, and the available ones, $y_k$, considering different  probabilities $\overline{\gamma}$ and $\overline{\lambda}$, and provides the filtering and fixed-point smoothing estimates of $x_k$ obtained from theorems 1 and 2, respectively.

\vskip .25cm
Considering one thousand independent simulations, each  with fifty iterations of the algorithms, in order to quantify the performance of
the proposed  estimators,  we use the root mean square error (RMSE) criterion, which is widely used because it allows  straightforward quantitative comparisons.
Denoting $\{x_k^{(s)}\}_{ k=1,\ldots,50}$  the $s$-th set of the simulated data (which is taken as the $s$-th set of true values of the
signal), and  $\widehat{x}^{(s)}_{k/k+h}$ as  the filtering ($h=0$) and fixed-point smoothing ($h=2$) estimates  at time $k$ in the $s$-th simulation run, the RMSE at time $k$ is calculated  by
$$\mbox{RMSE}_k =\sqrt{\displaystyle\frac{1}{1000} {\sum_{s=1}^{1000}}\left(
x_k^{(s)}-\widehat{x}_{k/k+h}^{(s)}\right)^{2}}, \ \ \ 1\leq k\leq 50, \ \ \ h=0,2.$$

First, considering fixed probabilities $\overline{\gamma}=0.7$ and  $\overline{\lambda}=0.3$, Figure 1 displays the values $\mbox{RMSE}_k$, for $k=1,\ldots,50$, corresponding to the filtering ($\widehat{x}_{k/k}$) and fixed-point smoothing ($\widehat{x}_{k/k+2}$) estimates.
From this figure, it can be seen that, at any time $k$, the $\mbox{RMSE}_k$ of the  smoothing estimates  is smaller than that  of the filtering estimates; hence, according to the $\mbox{RMSE}_k$ criterion, the smoother outperforms the filter.
Analogous results are obtained for other values of the  probabilities $\overline{\gamma}$ and  $\overline{\lambda}$.


In order to analyze the  overall performance   of the estimations provided by the  filtering and  smoothing algorithms  as a function of  the deception  attack probability $\overline{\lambda}$, Figure 2  shows   the mean values
of  $\mbox{RMSE}_k$ corresponding to the 50 iterations, versus $\overline{\lambda}= 0.1$ to 0.9,  when $\overline{\gamma}=0.7$ and $0.9$.
As expected, it is shown that the mean values of  $\mbox{RMSE}_k$, for both filtering and smoothing estimates,  become larger as the deception attack probability $\overline{\lambda}$ increases, and this increase is less
noticeable for high values of $\overline{\lambda}$.
 Furthermore, this figure shows
the  superiority of the smoother over the
filter and also that, for $\overline{\gamma}=0.9$, the results of both estimators are better than those obtained for $\overline{\gamma}=0.7$, and this improvement is more evident for values of $\overline{\lambda}$ less than or equal to 0.5.


Finally, to  illustrate the influence of $\overline{\gamma}$ (the probability that the observations contain the signal) on the performance of the estimators, Figure 3 shows   the mean values of  $\mbox{RMSE}_k$ corresponding to the 50 iterations, for a range of
$\overline{\gamma}$ values from 0.1 to 0.9,  when $\overline{\lambda}=0.1$ and $0.3$.
For these values of $\overline{\lambda}$, as expected, the mean values of  $\mbox{RMSE}_k$, for both filtering and smoothing estimates,  decrease  as the probability $\overline{\gamma}$ increases, meaning that better estimations  are obtained  as the probability that the signal is missing in the measurements, $1-\overline{\gamma}$, decreases. Moreover, this improvement is more noticeable for high values of $\overline{\gamma}$.
 This figure also shows
the  superiority of the smoother over the
filter and that, for $\overline{\lambda}=0.1$, the results of both estimators are better than those obtained for $\overline{\lambda}=0.3$, being this improvement more significant for values of $\overline{\gamma}$ greater than or equal to 0.5.


\section{Conclusions}
\label{sec:conclusions}

Recursive algorithms for the LS filtering and fixed-point smoothing problems from nonlinear observations perturbed by time-correlated additive noise and subject to random failures, causing that some measured outputs are only noise, have been proposed. The possibility of random deception attacks adds some complexity to the mathematical model considered. Using linear approximations of the actual observations, together with the projection theory and the innovation approach, the derivation of the estimation algorithms is based on the EKF approach. Some numerical results are used to examine the performance of the proposed estimators and to analyze the effect of missing measurement and deception attack success probabilities on the estimation accuracy.

Future research topics would include extending the proposed framework to deal with other nonlinear estimation approaches, such as unscented Kalman filtering, cubature Kalman filtering, particle filtering, Gaussian-Hermite filtering, divided differences filtering or Bayesian filtering. Another interesting further research direction would be considering nonlinear multisensor systems with different communication constraints (random delays, fading measurements and packet dropouts, among others) and different communication protocols (e.g., event triggering mechanisms, random communication protocol or round-robin protocol).

\section*{Disclosure statement}

The authors report there are no competing interests to declare.

\section*{Data availability statement}
Data sharing is not applicable to this article as no new data were created or analyzed in this study.

\section*{Funding}

Grant PID2021-124486NB-I00 funded by MCIN/AEI/ 10.13039/501100011033 and by ``ERDF A way of making Europe".

\section*{Notes on contributors}

\emph{R. Caballero-\'Aguila} received the M.Sc. degree in Mathematics and Ph.D. degree in Polynomial Filtering in Systems
with Uncertain Observations, both from the University of Granada (Spain), in 1997 and 1999, respectively. In 1997, she joined the University of Ja{\'e}n (Spain), where she is currently a Professor with the Department of Statistics and Operations Research.
Her current research interests include stochastic systems, filtering, prediction and smoothing.

\vskip .25cm

\noindent \emph{Jun Hu} received the B.Sc. degree in information and computing science and M.Sc. degree in Applied Mathematics from the Harbin University of Science and Technology, Harbin, China, in 2006 and 2009, respectively, and the Ph.D. degree in control science and engineering from the Harbin Institute of Technology, Harbin, in 2013. From 2010 to 2012, he was a Visiting Ph.D. Student with the Department of Information Systems and Computing, Brunel University, Uxbridge, U.K. From 2014 to 2016, he was an Alexander von Humboldt Research Fellow with the University of Kaiserslautern, Kaiserslautern, Germany. From 2018 to 2021, he was a Research Fellow with the University of South Wales, U.K. He is currently with the Department of Applied Mathematics, Harbin University of Science and Technology, and also with the School of Automation, Harbin University of Science and Technology, Harbin. His research interests include nonlinear control, filtering and fault estimation, time-varying systems and complex networks.

\vskip .25cm

\noindent \emph{J. Linares-P\'erez} received the M.Sc. degree in Mathematics and Ph.D. in
Stochastic Differential Equations, both from the University of Granada (Spain) in 1980
and 1982, respectively. She is currently a Professor with the Department of Statistics and
Operations Research, University of Granada (Spain). Her research interests are related with
stochastic calculus and estimation in stochastic systems.

\newpage

\renewcommand\thefigure{\arabic{figure}}

\begin{figure}[H]
\centering
\includegraphics[width=13 cm]{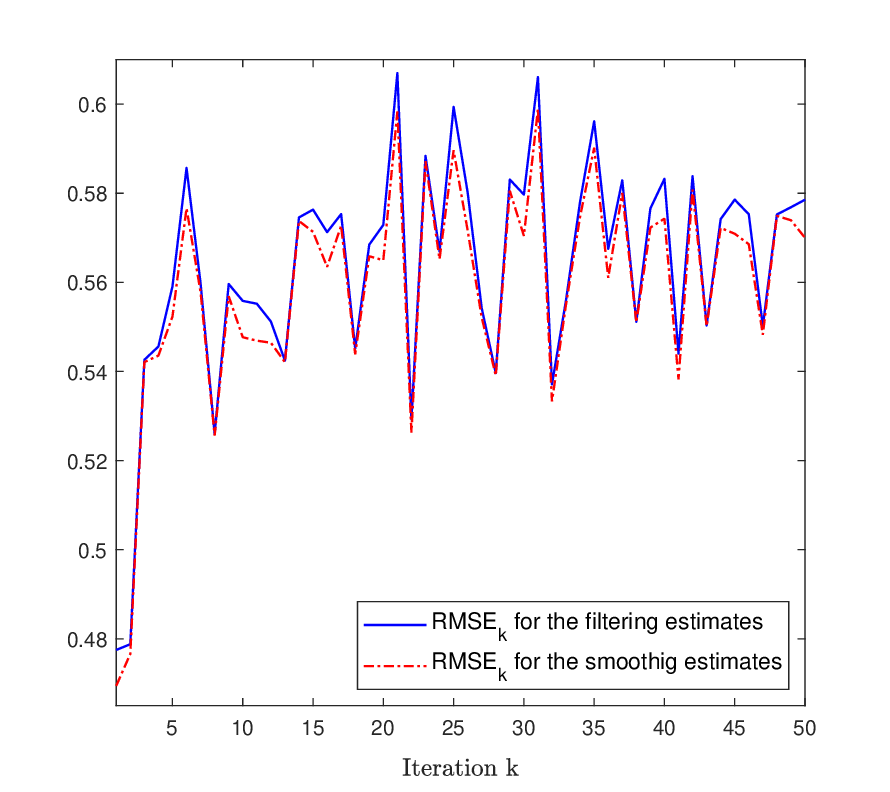}
 \caption{}
 \label{fig1}
\end{figure}

\begin{figure}[H]
 \centering
  \includegraphics[width=13 cm]{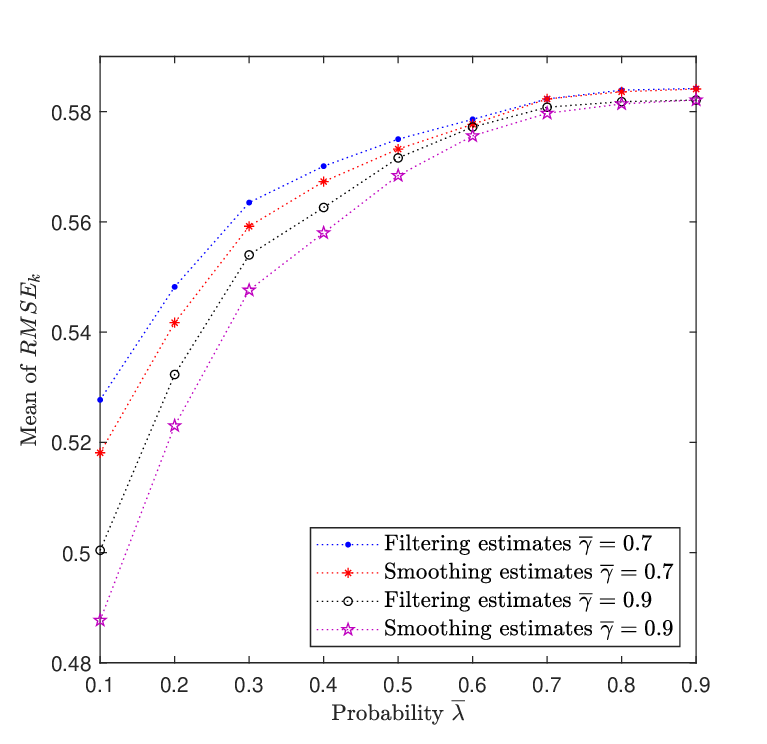}
  \caption{}
  \label{fig2}
\end{figure}

\begin{figure}[H]
 \centering
  \includegraphics[width=13 cm]{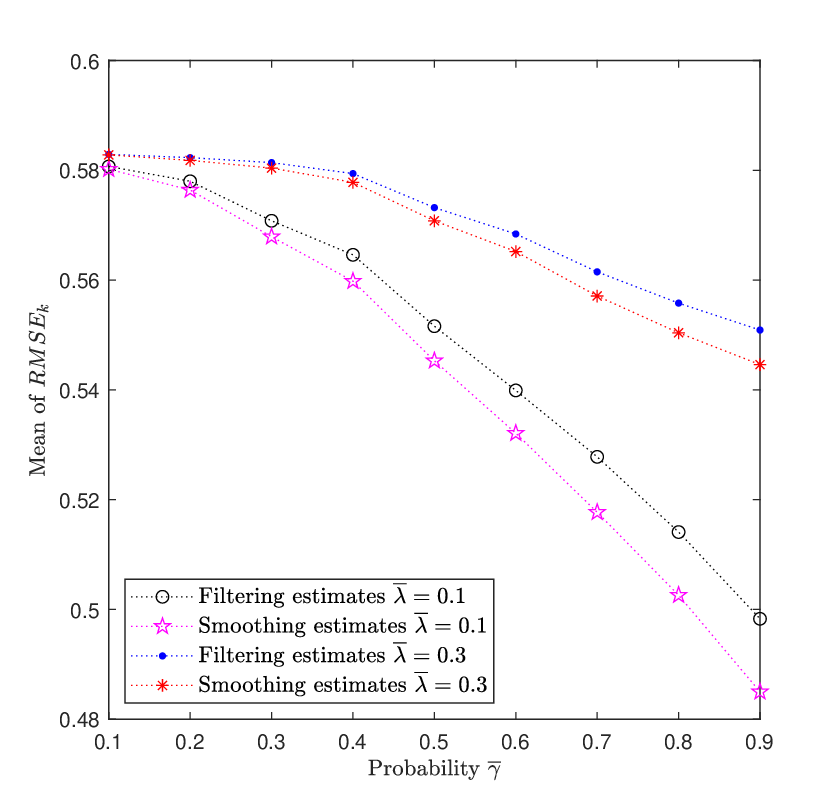}
  \caption{}
 \label{fig3}
\end{figure}

\newpage

\section*{Figure captions}

\noindent {\bf Figure 1.} $\mbox{RMSE}_k$ for the filtering and smoothing   estimates,  when $\overline{\gamma}=0.7$ and  $\overline{\lambda}=0.3$.

\vskip.25cm

\noindent {\bf Figure 2.}
Means of $\mbox{RMSE}_k$ for the filtering and smoothing   estimates  versus $\overline{\lambda}$, when $\overline{\gamma}=0.7$ and 0.9.

\vskip.25cm

\noindent {\bf Figure 3.}
Means of $\mbox{RMSE}_k$ for the filtering and smoothing   estimates versus $\overline{\gamma}$, when $\overline{\lambda}=0.1$ and 0.3.

\end{document}